\documentclass{article}
\usepackage{amsmath,graphicx}

\usepackage{microtype}
\usepackage{subfigure}
\usepackage{booktabs} 
\usepackage[matit]{echodefs}
\usepackage{lipsum}
\usepackage{amsfonts}
\usepackage{epstopdf}
\usepackage{mathtools}
\usepackage[autostyle]{csquotes}
\usepackage{amsmath,bm}
\usepackage{cleveref}
\usepackage{thmtools, thm-restate}
\usepackage{romannum}
\usepackage{geometry}

\usepackage{algorithm}
\usepackage{algorithmic}
\usepackage{authblk}

\title{Real Polynomial Gram Matrices Without Real Spectral Factors}
\author{Puoya Tabaghi and Ivan Dokmani\'c}
\affil{Coordinated Science Laboratory, ECE \\ University of Illinois at Urbana-Champaign, IL, USA\\ \{tabaghi2, dokmanic\}@illinois.edu}

\date{}

\begin{document}
%
\maketitle
\begin{abstract}
It is well known that a non-negative definite polynomial matrix (a polynomial Gramian) $G(t)$ can be written as a product of its polynomial spectral factors, $G(t) = X(t)^H X(t)$. In this paper, we give a new algebraic characterization of spectral factors when $G(t)$ is real-valued. The key idea is to construct a representation set that is in bijection with the set of real polynomial Gramians. We use the derived characterization to identify the set of all complex polynomial matrices that generate real-valued Gramians, and we formulate a conjecture that typical rank-deficient real polynomial Gramians have real spectral factors.
\end{abstract}

\section{Introduction}
Spectral factorization was introduced by Wiener \cite{wiener1949extrapolation,wiener1957prediction}, 
and has since found numerous applications in signal processing, communications, optimal control, and network theory \cite{fischer2005sorted, hunt1993polynomial, kailath2000linear, newcomb1966linear}. Most  of the literature addresses spectral factorization of non-negative polynomial matrices in the complex plane \cite{hardin2004matrix, anderson1967algebraic, jevzek1985efficient, callier1985polynomial, sayed2001survey}. An elegant proof of the existence and uniqueness of spectral factorization for full-rank polynomial matrices on the unit circle appears in \cite{ephremidze2009simple} with an extension to the the rank-deficient polynomials in \cite{ephremidze2015rank}. A factorization theorem for full-rank polynomial matrices factorization in the complex plane can also be extended to polynomial matrices on the real line \cite{ephremidze2014elementary}; another elegant proof is given in \cite{van2014straightforward}. Arguments similar to those in \cite{ephremidze2009simple,ephremidze2015rank} can be made for rank-deficient polynomial matrices on the real line. 

In this paper, we find a common structure of real non-negative polynomial matrices on the real line. 
We first show that to every real polynomial Gramian we can uniquely associate a certain block-Toeplitz bilinear system. The parameters of this bilinear system form an alternative representation of the original matrix which makes it possible to characterize the set of real polynomial Gram matrices without real spectral factors. 

The motivation for this work comes from localization of moving points using time-dependent pairwise distance measurements. As shown in \cite{tabaghi2018kinetic}, under a polynomial trajectory model, localization can be cast a semidefinite program followed by a factorization of a polynomial Gram matrix. In order for the resulting trajectories to be realizable in the 3D Euclidean space, the factors (that is, the coordinates) should be real. We conclude the paper by a conjecture that a \textit{typical} real polynomial Gramian has a real spectral factor, provided that is sufficintly rank deficient.

\subsection{Notation}
In this paper, $\R$ (resp. $\C$) denotes the field of real (resp., complex) numbers, $\R^{d \times N} $ (resp. $\C^{d \times N}$) is the ring of real (resp. complex) $d \times N$ matrices, $\mathbb{U}(d)$ is the group of unitary $d \times d$ matrices, and $\R^{d \times N}[t]$ is the ring of $d \times N$ polynomial matrices with real coefficients on the real line, that is, $t \in \R$. If $X(t) = \sum_{p=0}^{P}{A_p t^p} \in \C^{d\times N}[t]$, then $X(t)^H = \sum_{p=0}^{P}{A_p^{H}t^p}$ where $A_p^{H}$ is the Hermitian transpose of $A_p$. 

\section{Existence and uniqueness of spectral factorization} 
\label{sec:existence_and_uniqueness_of_spectral_factorization}

We begin by briefly reviewing several known facts about spectral factorization that we will need. Derivations and detailed discussions about those results can be found in \cite{ephremidze2015rank, ephremidze2014elementary}

\begin{itemize}
    \item Fact 1: Given a rank $d$ non-negative definite polynomial matrix $G(t) \in \R^{N \times N}[t]$, there exists $X(t) \in \C^{d \times N}[t]$ such that $G(t) = X(t)^H X(t)$; we call $X(t)$ a spectral factor of $G(t)$. 
    \item Fact 2: Clearly $U X(t)$ is also a spectral factor for any unitary $U \in \mathbb{U}(d)$.
    \item Fact 3: The set $\{ U X(t) \ : \ U \in \mathbb{U}(d) \}$ contains \emph{all} spectral factors of $G(t)$; there are no other spectral factors that have a different form. In particular, $U$ cannot be a function of $t$. 
    \item Fact 4: For any real Gramian $G(t)$, there exists a unique factorization $X(t) = \sum (R_p + i Q_p)t^p$ such that $Q_0 = 0$ and $R_0$ has orthogonal rows with the first non-zero entry in each row being positive. If $G(t)$ has a real factorization, then it has a unique real factorization with $R_0$ and $Q_0$ as above.
\end{itemize}

\subsection{Differences between polynomial and constant Gram matrices} 

For constant real positive semidefinite matrices, we can always get a real spectral decomposition simply by the eigenvalue decomposition. Namely, given a Gramian $G \in \R^{N \times N}$, there exist matrices matrices $U, V \in \R^{N \times N}$, such that $U U^T = I$, $V$ is diagonal, and
\[
    G = U V U^T.
\]
Of course, there exists an infinitude of complex factorizations---for any unitary $Q \in \C^{N \times N}$, it holds that $G = (UQ) V (UQ)^H$, but the existence of a real one is guaranteed.

The situation is fundamentally different in the polynomial case. While there are still infinitely many unitarily-related factorizations of any polynomial Gramian, there exist \emph{real} polynomial positive-semidefinite matrices $G(t)$ such that none of these factorizations is real. The intention of this paper is to characterize all such matrices by deriving a certain generative model.


\section{Main results} \label{sec:main_result}

Our first result is a characterization of spectral factors of real polynomial Gramians. More precisely, we want to
\begin{align*}
& \text{find}       & &  A_0, \cdots, A_P \in \C^{d \times N} & \\
& \text{such that} & & X(t)^{H}X(t) \in \R^{N \times N}[t] \\
& \text{where} & & X(t)=\sum_{p=0}^{P}{A_p t^p}.
\end{align*}
We start by proving that these spectral factors must satisfy a certain block-Toeplitz system of equations. Throughout the paper we let $\mathcal{G}$ be the set of positive semidefinite polynomial matrices of the form
\begin{equation*}
G(t) = \sum_{k=0}^{K}{B_k t^{k}}, ~ t \in \R,
\end{equation*}
where $B_k \in \R^{N \times N}$, $K=2P$ and $\max_{t \in \R}\rank G(t) = \rank G(0) = d$. The rank condition is a minor carveout that lets us do away with a number of technicalities. In particular, it implies that $\rank B_0 = d$. It is convenient to eliminate the unitary ambiguity by fixing $U$ through an eigendecomposition of $B_0$, as we do next.
\begin{lemma}
\label{lem:unique-orth}
For any $G(t) \in \mathcal{G}$, there is a unique polynomial matrix 	$X(t)$ such that $G(t) = X(t)^H X(t)$, where 
\[
    X(t) = \sum_{p=0}^{P}{A_p t^{p}},
\]
$A_p = R_p + i Q_p$, $R_p, Q_p \in \R^{d \times N}$, $Q_0 = 0$ and $R_0 \in \mathcal{X}_0$. The set $\mathcal{X}_0$ is the set of $d \times N$ matrices with orthogonal rows, such that the first non-zero entry in each row is positive.
\end{lemma}

\begin{proof}
From Fact 1, $G(t)$ has a $d \times N$ spectral factor $X(t) = \sum_{p=0}^P A_p t^p$ such that $G(t) = X(t)^H X(t)$. Since $G(t)$ is real, $B_0 = A_0^T A_0$ is real as well, and we can choose $X(t)$ so that $A_0$ is real. This gives $Q_0 = 0$.

If $\tilde{X}(t)$ is also a spectral factor of $G(t)$ then necessarily $\tilde{X}(t) = U X(t)$ for some (constant in $t$) unitary matrix $U$. Since $Q_0 = 0$, all admissible $U$ are real. Let $V \Lambda V^T$ be an eigenvalue decomposition of the real matrix $R_0 R_0^T$ with both the eigenvalues and eigenvectors real. Any matrix that diagonalizes $R_0 R_0^T$ (that is, makes the rows of $R_0$ orthogonal) has the form $\tilde{V} = \diag (s) V^T$ with $s \in \{ -1, 1 \}^d$ since $\tilde{V} R_0 (\tilde{V} R_0)^T = \Lambda$. Asking, in addition to orthogonality, that the first non-zero entry in each row be positive uniquely determines the sign sequence $s$. In other words, among all spectral factors $U X(t)$ of $G$, there is exactly one such that $U R_0 \in \mathcal{X}_0$.
\end{proof}

\Cref{lem:unique-orth} establishes a unique spectral factorization map for real $G(t)$. We can now state our first main result which is proved in \cref{app:1}.

\begin{theorem} \label{thm:poly}
With notation as in \cref{lem:unique-orth}, for any $G(t) \in \mathcal{G}$, there exist unique matrices $W_k \in \R^{d \times d}$, $k \in \{1,\cdots, 2P\}$ and $R_p \in \R^{d \times N}$, $p \in \{ 0, \ldots, P \}$, with $W_k$ symmetric for every $k$ and $R_0 \in \mathcal{X}_0$, such that
\begin{equation} \label{eq:A_WX}
\begin{bmatrix}
    A_{0} \\
    A_{1}  \\
    \vdots \\
    A_{P} 
\end{bmatrix} = 
\begin{bmatrix}
    I               &   0               & \cdots    & 0             & 0 \\
    i W_{1} &   I               & \cdots    & 0             & 0 \\
    i W_{2} &   i W_{1}     & \cdots    & 0             & 0 \\
    \vdots      & \vdots        & \vdots    & \ddots    & 0 \\
    i W_{P} & i W_{P-1} & \cdots    & i W_{1} & I
\end{bmatrix}
\begin{bmatrix}
    R_{0} \\
    R_{1}  \\
    \vdots \\
    R_{P} 
\end{bmatrix},
\end{equation}
and
\begin{equation} \label{eq:WX}
\underbrace{\begin{bmatrix}
    W_{P+1} &   W_{P}           & W_{P-1}       & \cdots    & W_{1} \\
    W_{P+2} &   W_{P+1}         & W_{P}         & \cdots    & W_{2} \\
    W_{P+3} &   W_{P+2}         & W_{P+1}       & \cdots    & W_{3} \\
    \vdots      & \vdots            & \vdots            & \ddots    & \vdots     \\
    W_{2P}  & W_{2P-1}  & W_{2P-2}  & \cdots    & W_{P}
\end{bmatrix}}_{:=W}
\underbrace{\begin{bmatrix}
    R_{0} \\
    R_{1}  \\
    \vdots \\
    R_{P} 
\end{bmatrix}}_{:= R}
= 0
\end{equation} 
and $\sum_p A_p t^p$ is a spectral factor of $G(t)$.

Denote further by $( \{ W_k \}, \{ R_p\})$  a collection of matrices $W_k$ and $R_p$ as above, and define
\[
    \mathcal{H} := \big\{ ( \{ W_k \}, \{ R_p\}) \ : \ WR = 0~\text{and}~R_0 \in \mathcal{X}_0 \big\},
\]
with $W$ and $R$ defined as in \cref{eq:WX}. Then for every $G(t) \in \mathcal{G}$, there exists exactly one $( \{ W_k \}, \{ R_p\}) \in \mathcal{H}$ satisfying \cref{eq:A_WX} and \cref{eq:WX} and the mapping $G(t) \mapsto \phi(G(t)) = ( \{ W_k \}, \{ R_p\})$ is a bijection between $\mathcal{G}$ and $\mathcal{H}$.
\end{theorem}
\begin{remark}
Theorem 3.2 gives us a means to generate complex polynomial matrices with real Gramians, a task that might seem nontrivial a priori. With \cref{thm:poly} in hand, one simply has to produce $2P$ symmetric matrices $W_k$ and a suitable $R$ in the nullspace of the obtained block-Toeplitz matrix $W$. This is always possible since $W$ has at least a $d$-dimensional nullspace.
\end{remark}

We have now identified the special structure of the coefficients of complex polynomial matrices that have real-valued Gramians. It seems intuitive that only particular polynomial matrices have the required property. The fact that there is a bijection between $\mathcal{G}$ and $\mathcal{H}$ allows us to represent any $G(t) \in \mathcal{G}$ in terms of $( \{ W_p \}, \{ R_p\})$.

The above representation is useful because it leads to a simple characterization of the set of real positive semidefinite polynomial matrices in $\mathcal{G}$ without real spectral factors. Denote the set of such polynomial matrices by $\mathcal{G}_C \subset \mathcal{G}$. Let further $\phi : \mathcal{G} \to \mathcal{H}$ be the the above bijective map that associates $( \{ W_p \}, \{ R_p\})$ to every $G \in \mathcal{G}$. We then have the following result:

\begin{lemma} \label{lem:image-set}
Let $\mathcal{G}_{C} \subset \mathcal{G}$ be the set of real polynomial Gramians with no real spectral factor. Then, $\phi(\mathcal{G_{C}}) = \{ (\{ W_k \}, \{ R_p\}) \in \mathcal{H}: W \neq 0 \} $.
\end{lemma}
\begin{proof}
If $G(t) \in \mathcal{G}_{R} = \mathcal{G} \setminus \mathcal{G}_{C}$, then the coefficients of its spectral factor are real valued. Since the first coefficient, $R_0$, is full rank, using first \eqref{eq:A_WX} and then \eqref{eq:WX} gives that $W_k = 0$ for all $k \in \{1, \cdots, 2P\}$. The claim follows by contraposition.
\end{proof}

\section{A Conjecture about Real Gram Matrices}

A useful application of this representation would be to prove that (in some sense) typical rank-deficient real polynomial Gramians have real spectral factors. This idea stems from the intuition that a complex polynomial matrix, $X(t) \in \C^{d \times N}[t]$, with a real Gramian, $X(t)^{H} X(t) \in \R^{N \times N}[t]$, can be characterized by $N^{2}$ independent polynomial equations with $O(d N)$ parameters. This could imply that a large enough generic real polynomial Gramian (that is to say, for a large $N$) has a real spectral factor.

While we could not produce a proof of such a result, we state here a conjecture based on studying the problem ``by hand'' for small values of $N$ and $d$. The general case might involve a study of the solution set of a homogeneous system of bilinear equations \eqref{eq:WX}. Concretely, we would like to know whether there exist choices of $N$ and $d$ such that the size of $\mathcal{G}_R = \mathcal{G} \setminus \mathcal{G}_{C}$ is larger than $\mathcal{G}_{C}$. With the bijection between $\mathcal{G}$ and $\mathcal{H}$ established in \cref{thm:poly}, one route would be via images of $\mathcal{G}_R$ and $\mathcal{G}_C$ through $\phi$. Since it is possible to define a measure on $\mathcal{H}$, we also have the pullback measure (through $\phi$) on $\mathcal{G}$. The goal would then be to show that $\mu(\phi(\mathcal{G}_C)) = 0$ as soon as $N$ is sufficiently large (or $d$ sufficiently small), as stated below.

\begin{conjecture}
If $N > d+1$, the set of $G(t) \in \mathcal{G}$ with no real spectral factor has measure zero in $\mathcal{G}$.
\end{conjecture}

\section{Conclusion}
We derived a generative representation of the set of real non-negative polynomial matrices on the real line, which do not have any real spectral factors. One application of this representation is to generate random rank-deficient real polynomial Gramians by randomly sampling $( \{ W_k \}, \{ R_p\}) \in \mathcal{H}$. This can effectively be done by generating a set of random symmetric matrices, $\{W_k\}$, and constructing the block-Toeplitz matrix $W$. Then, any matrix $R \in \mathcal{N}(W)$, along with \eqref{eq:A_WX} yields the coefficients of a spectral factor of $G(t) = X(t)^H X(t)$ that is a real-valued polynomial Gramian. We conjecture that this representation can be used study the size of the subsets of $\mathcal{G}$, and in particular, that the set of real rank-deficient polynomial Gramians without real spectral factors is in some sense small.

\appendix

\section{Key technical lemma}

The key result we use several times is a particular form of a solution of the matrix equation $X^{T}A - A^{T}X = C$. While solutions to this equation are already known in the literature \cite{braden1998equations}, we give a simpler form which fits our purpose better.
\begin{lemma}\label{lem:main}
Let $R_0$ be such that 
\begin{equation*}
R_0^{T}A - A^{T}R_0 = C
\end{equation*}
where $R_0, A_0 \in \R^{d \times N}, C \in \R^{N \times N}$. Then any solution to 
\begin{equation}\label{eq:affine}
X^{T}A - A^{T}X = C
\end{equation}
can be written as
\begin{equation} \label{eq:affinesol}
X = W A + R_0
\end{equation}
for some symmetric $W \in \R^{d \times d}$.
\end{lemma}
\begin{proof}
Let us assume that \eqref{eq:affine} has a particular solution, $R_0$. One can simply write 
\begin{equation*}
X = W A + W^{'} A^{\perp}
\end{equation*}
where $W \in \R^{d \times d}, W^{'} \in \R^{(N-d)\times (N-d)}$ and the rows of $A^{\perp} \in \R^{(N-d)\times N}$ span the the orthogonal complement of the row space of $A$. If $N=d$, then $W^{'}=0$. The homogeneous solution to \eqref{eq:affine} is
\begin{align*}
0 &= X^{T} A - A^{T} X \\
&=(W A + W^{'} A^{\perp})^{T} A - A^{T}(W A + W^{'} A^{\perp}) \\
&=(W A + W^{'} A^{\perp}-W^{T} A)^{T} A - A^{T} W^{'} A^{\perp} 
\end{align*}
where $N \neq d$. The rows of each term lie in  $A$ and $A^{\perp}$ row spaces which are orthogonal complement of each other. Since $A$ has full row rank, $W^{'} = 0$. The trivially holds for the case where $N=d$. Thus,
\begin{equation*}
A^{T}(W-W^{T})A = 0
\end{equation*}
which leads to $W =W^{T}$. Therefore, \eqref{eq:affinesol} generally solves \eqref{eq:affine} for any symmetric matrix $W$ if \eqref{eq:affine} admits a particular solution, $R_0$. 
\end{proof}

\section{Proof of Theorem \ref{thm:poly}}\label{app:1}
Since $G(t) \succeq 0$ for almost all $t \in \R$ and $\max_{t \in \R} \rank{G(t)} = d$, there exists a unique, up to a constant $d \times d$ left unitary factor, polynomial matrix $X(t) \in \C^{d\times N}[t]$ of degree $P$ such that $G(t) = X(t)^{H} X(t)$.
Letting $X(t) = \sum_{p=0}^{P}{A_p t^{p}} \in \C^{d \times N}[t]$ and $G(t) = \sum_{k=0}^{2P} B_k t^k$, we get
\begin{equation*}
B_k = \sum_{\substack { 0 \leq p \leq P \\ 0 \leq k-p \leq P}}{A_{p}^{H} A_{k-p}} \in \R^{N \times N}.
\end{equation*}
By the definition of $\mathcal{G}$ and \cref{lem:unique-orth}, we know that $G(t)$ has a spectral factor $X(t)$ such that $A_0 = R_0$, $Q_0 = 0$, and $R_0 \in \mathcal{X}_0$ has orthogonal rows hence full row rank.

Let us use induction on $p$ to show that
\begin{equation} \label{eq:indHyp1}
Q_p = \sum_{i=1}^{p}{W_i R_{p-i}}
\end{equation}
holds for some symmetric matrices $W_i \in \R^{d \times d}$, $\forall i, p \in \{1, \cdots, P\}$. We first establish the induction base ($p=1$) by setting the imaginary part of $B_1$ to zero, i.e. $\mathcal{I}(B_1) = 0$, which is equivalent to 
\begin{equation*}
R_0^{T} Q_1-Q_1^{T} R_0  = 0 .
\end{equation*}
Since $\rank{R_0} = d$ ($\rank{G(0)} = d$), \cref{lem:main} yields
\begin{equation*}
Q_1 = W_1 R_0
\end{equation*}
for any symmetric matrix $W_1 \in \R^{d \times d}$. We assume that the induction hypothesis \eqref{eq:indHyp1} holds for $p' < P$, and prove that it also holds for $p = p'+1$. We claim that the general solution to $\mathcal{I}(B_{p}) = 0$, has the form of \eqref{eq:indHyp1}. Setting $\mathcal{I}(B_{p}) = 0$ leads to
\begin{align*}
\mathcal{I}(B_{p}) &\stackrel{(a)}{=} R_0^{T}Q_p - Q_p^{T} R_0 + \sum_{i=1}^{p-1}{ R_i^{T} Q_{p-i}} - \sum_{i=1}^{p-1}{Q_{p-i}^{T} R_i}  \\
&\stackrel{(b)}{=} R_0^{T}Q_p - Q_p^{T} R_0 + \sum_{i=1}^{p-1}{\sum_{j=1}^{p-i}{ R_i^{T}W_j R_{p-i-j}}} - \sum_{i=1}^{p-1}{\sum_{j=1}^{p-i}{R_{p-i-j}^{T}W_j }R_i}  \\
&\stackrel{(c)}{=} R_0^{T}(Q_p- \sum_{i=1}^{p-1}{W_{p-i} R_i}) - (Q_p-\sum_{i=1}^{p-1}{W_{p-i}R_i})^{T} R_0 .
\end{align*}
where $(a)$ follows from 
\begin{equation*}
\mathcal{I}(B_{p}) = \sum_{i=0}^{p-1}{R_i^{T} Q_{p-i}} - \sum_{i=0}^{p-1}{Q_{p-i}^{T} R_i},
\end{equation*}
$(b)$ from induction hypothesis \eqref{eq:indHyp1} and $(c)$ from
\begin{equation*}
\sum_{i=1}^{p-1}{\sum_{j=1}^{p-i}{ R_i^{T}W_j R_{p-i-j}}} = \sum_{i=1}^{p-1}{\sum_{j=1}^{p-i-1}{ R_i^{T}W_{j} R_{p-i-j}}} + (\sum_{i=1}^{p-1}{W_{p-i}R_i})^{T} R_{0}
\end{equation*}
and $\sum_{i=1}^{p-1}{\sum_{j=1}^{p-i-1}{ R_i^{T}W_{j} R_{p-i-j}}}$ being a symmetric matrix. Therefore, the general solution to $\mathcal{I}(B_{p}) = 0$ is
\begin{equation*}
Q_{p} =\sum_{i=1}^{p-1}{ W_i R_{p-i}} + W_{p} R_0
\end{equation*}
for any symmetric matrix $W_{p} \in \R^{d \times d}$ (refer to \cref{lem:main}). This completes the proof of induction for $1< p \leq P$. Consequently,
\begin{equation*}
\begin{bmatrix}
    A_{0} \\
    A_{1}  \\
    \vdots \\
    A_{P} 
\end{bmatrix} =
\begin{bmatrix}
    R_{0} \\
    R_{1}  \\
    \vdots \\
    R_{P} 
\end{bmatrix} +
j\begin{bmatrix}
    0				& 	0 				& \cdots 	& 0 			& 0	\\
    W_{1}		& 	0				& \cdots 	& 0 			& 0	\\
    W_{2}		&  W_{1} 		& \cdots 	& 0 			& 0	\\
    \vdots		& \vdots		& \vdots	& \ddots	& \vdots	\\
    W_{P}		& W_{P-1} 	& \cdots 	& W_{1} 	& 0
\end{bmatrix}
\begin{bmatrix}
    R_{0} \\
    R_{1}  \\
    \vdots \\
    R_{P} 
\end{bmatrix}
\end{equation*}
which concludes the first part of \cref{thm:poly}. The second part is due to  $\mathcal{I}(B_{p+P})= 0$ constraints for $p \in \{1,\cdots, P\}$. Using induction, we prove that the following holds for $p \in \{1,\cdots, P\}$,
\begin{equation}\label{eq:indHyp2}
\sum_{i=0}^{P}{ W_{P+p-i} R_i} = 0,
\end{equation}
where $W_p \in \R^{d \times d}$ and $W_p = W_p^{T}$, $\forall p \in \{1,\cdots, 2P\}$. The induction base, $p=1$, can be established by analyzing $\mathcal{I}(B_{P+1}) = 0$ equation
\begin{equation*}
\mathcal{I}(B_{P+1}) = C_{P+1} - C_{P+1}^{T}
\end{equation*}
where $C_{P+1} = R_{1}^{T} Q_{P} + R_{2}^{T} Q_{P-1}+ \cdots + R_{P}^{T}Q_{1}$. With the results established in the previous part, $C_{P+1}$ can written as
\[ C_{P+1} = 
R^{T}
\begin{bmatrix}
    0 				& 	0 				& \cdots 	& 0 			& 0 			& 0 \\
    W_P			& 	W_{P-1}	& \cdots	& W_{2} 	& W_{1}	& 0 \\
    W_{P-1}	& 	W_{P-2}	& \cdots 	& W_{1}	& 0 			& 0 \\
    \vdots		& \vdots		& \vdots	& \vdots	& \vdots	& \vdots	 \\
    W_{1}		& 0 				& \cdots 	& 0			& 0 			& 0
\end{bmatrix}
R
\]
where $R^{T} = \begin{bmatrix} R_{0}^{T}, R_{1}^{T}, \cdots, R_{P}^{T} \end{bmatrix}$. Thus, $\mathcal{I}(B_{P+1})$ can be simplified to
\[
\mathcal{I}(B_{P+1}) = 
R^{T}
\begin{bmatrix}
    0 				& 	-W_P 	& \cdots 		& -W_{3}	& -W_2 	& -W_1 \\
    W_P			& 	0	 		& \cdots	 	& 0 			& 0			& 0 \\
    W_{P-1}	& 	0			& \cdots 		& 0 			& 0 			& 0 \\
    \vdots		& \vdots	& \vdots		& \vdots	& \vdots 	& \vdots	 \\
    W_{1}		& 0 			& \cdots 		& 0			& 0 			& 0
\end{bmatrix}
R
 \]
or equivalently
\begin{equation*}
\mathcal{I}(B_{P+1}) = (W_{P} R_{1} + \cdots + W_{1} R_{P})^{T} R_0 - R_0^{T}(W_{P} R_{1} + \cdots + W_{1} R_{P}) .
\end{equation*}
\cref{lem:main} guarantees that there exists a symmetric matrix $W_{P+1} \in \R^{d \times d}$ such that
\begin{equation*}
 W_{P+1} R_{0} + W_{P}  R_{1} + \cdots +  W_{1}  R_{P} = 0 .
\end{equation*}
Now, we assume that the induction hypothesis \eqref{eq:indHyp2} holds for all $p^{'} < P$, and the general solution to  $\mathcal{I}(B_{P+p}) = 0$ for $p=p^{'}+1$ follows the hypothesis form. With some elaboration, $\mathcal{I}(B_{p+P})$ can be simplified to 
\begin{equation*}
\mathcal{I}(B_{p+P})= C_{P+p} - C_{P+p}^{T}
\end{equation*}
where $C_{P+p} = \sum_{j=0}^{p-1}{(\sum_{i=p}^{P}{W_{P+p-i-j} R_{i}})^{T} R_{j} }$. The induction hypothesis \eqref{eq:indHyp2} simplifies $C_{P+p}$ to,
\begin{align*}
C_{P+p} &= (\sum_{i=p}^{P}{W_{P+p-i} R_{i}})^{T} R_{0}  + \sum_{j=1}^{p-1}{(\sum_{i=p}^{P}{W_{P+p-i-j} R_{i}})^{T} R_{j} } \\
&= (\sum_{i=p}^{P}{W_{P+p-i} R_{i}})^{T} R_{0}  - \sum_{j=1}^{p-1}{(\sum_{i=0}^{p-1}{W_{P+p-i-j} R_{i}})^{T} R_{j} }.
\end{align*}
Pairing of terms with similar sign leads to 
\begin{equation*}
\mathcal{I}(B_{P+p}) = D_{P+p}- D_{P+p}^{T}
\end{equation*}
where 
\begin{equation*}
D_{P+p}= (\sum_{i=1}^{P}{W_{P+p-i} R_{i}})^{T} R_{0}  + \sum_{j=1}^{p-1}{\sum_{i=1}^{p-1}{R_{j}^{T}W_{P+p-i-j} R_{i}} }
\end{equation*}
Since $\sum_{j=1}^{p-1}{\sum_{i=1}^{p-1}{R_{i}^{T}W_{P+p-i-j} R_{j}} } $ is a symmetric matrix, $\mathcal{I}(B_{P+p})$ can be ultimately simplified to
\begin{equation*}
\mathcal{I}(B_{P+p}) = (\sum_{i=1}^{P}{W_{P+p-i} R_{i}})^{T} R_{0}- ((\sum_{i=1}^{P}{W_{P+p-i} R_{i}})^{T} R_{0})^{T} .
\end{equation*}
Finally, $\mathcal{I}(B_{P+p}) = 0$ is equivalent to,
\begin{equation*}
\sum_{i=0}^{P}{W_{P+p-i} R_{i}} = 0
\end{equation*}
for some symmetric matrix $W_{P+p}$. This completes the proof of the induction and establishes the existence of the representation of $G(t) \in \mathcal{G}$ in terms of $(\{W_k\}, \{R_p\}) \in \mathcal{H}$. 

To show that this representation is unique, first note that by \cref{lem:unique-orth}, $G(t)$ has a unique spectral factorization $X(t) = \sum A_p t^p$ in $\mathcal{A}$, where $\mathcal{A} = \{ \sum_p (R_p + i Q_p)t^p \in \C^{d \times N}[t] \ : \ Q_0 = 0, R_0 \in \mathcal{X}_0 \}$. Further, every element of $\mathcal{A}$ corresponds to a unique element of $\mathcal{G}$. Hence, the spectral factorization map is a bijection between $\mathcal{A}$ and $\mathcal{G}$. 

Now, let us show that there is a bijection between $\mathcal{A}$ and $\mathcal{H}$. Suppose that $(\{W_p\}, \{R_p\})$ and $(\{W_p'\}, \{R_p'\})$ both generate the same $X(t) \in \mathcal{A}$. Since the real parts coincide we have $R_p = R_p'$ for all $p$. Then from \cref{eq:A_WX} we get
\[
\begin{bmatrix}
    0           &   0               & \cdots    & 0             & 0 \\
    W_{1}-W^{'}_{1}     & 0             & \cdots    & 0             & 0 \\
    W_{2}-W^{'}_{2}     & W_{1}-W^{'}_{1}       & \cdots    & 0             & 0 \\
    \vdots      & \vdots        & \ddots    & \vdots    & \vdots    \\
    W_{P}-W^{'}_{P}     &W_{P-1}-W^{'}_{P-1}    & \cdots    & W_{1}-W^{'}_{1}   & 0
\end{bmatrix}
\begin{bmatrix}
    R_{0} \\
    R_{1}  \\
    \vdots \\
    R_{P} 
\end{bmatrix}
= 0.
\]
Since $R_0$ has full rank row rank, it must be that $\{ W_k \}_{k=1}^{P} = \{ W^{'}_k \}_{k=1}^{P}$. Finally, using this fact in \cref{eq:WX}, we deduce that $\{ W_k \}_{k=p+1}^{2P} = \{ W^{'}_k \}_{k=P+1}^{2P}$ (injectivity). Since by the first part of the proof, every $X(t) \in \mathcal{A}_0$ can be represented via some $(\{ W_k \}, \{ R_p \}) \in \mathcal{H}$ (surjectivity), we also have a bijection between $\mathcal{A}$ and $\mathcal{H}$. By composition we have a bijection between $\mathcal{G}$ and $\mathcal{H}$.

\bibliographystyle{IEEEbib-init}
\bibliography{sf}

\end{document}